\documentclass[%
 aip,
 amsmath,amssymb,
 reprint,%
]{revtex4-1}

\usepackage{graphicx}% Include figure files
\usepackage{dcolumn}% Align table columns on decimal point
\usepackage{bm}% bold math
\usepackage[utf8]{inputenc}
\usepackage[T1]{fontenc}
\usepackage{mathptmx}
\usepackage{etoolbox}
\usepackage{amsthm}
\usepackage{epstopdf}
\usepackage{subfigure}

\numberwithin{equation}{section}
\newtheorem{teor}{Theorem}[section]
\newtheorem{obs}[teor]{Remark}

\makeatletter
\def\@email#1#2{%
 \endgroup
 \patchcmd{\titleblock@produce}
  {\frontmatter@RRAPformat}
  {\frontmatter@RRAPformat{\produce@RRAP{*#1\href{mailto:#2}{#2}}}\frontmatter@RRAPformat}
  {}{}
}%
\makeatother
\begin{document}

\preprint{AIP/123-QED}

\title[Topological suppression of ES]{Topologically-induced suppression of explosive synchronization}
% Force line breaks with \\
\author{Manuel Miranda}
 \affiliation{Institute of Cross-Disciplinary Physics and Complex Systems, IFISC (UIB-CSIC), 07122 Palma de Mallorca, Spain}%Lines break automatically or can be forced with \\
\author{Mattia Frasca}%
\affiliation{ 
Department of Electrical, Electronics and Computer Science Engineering, University of Catania, Italy %\\This line break forced with \textbackslash\textbackslash
}%
\affiliation{Istituto di Analisi dei Sistemi ed Informatica “A. Ruberti”,
Consiglio Nazionale delle Ricerche (IASI-CNR), Roma, Italy}

\author{Ernesto Estrada}
  \email{estrada@ifisc.uib-csic.es}
\affiliation{%
Institute of Cross-Disciplinary Physics and Complex Systems, IFISC (UIB-CSIC), 07122 Palma de Mallorca, Spain%\\This line break forced% with \\
}%

\date{\today}% It is always \today, today,

% \begin{abstract}
% While phase oscillators coupled into a complex network typically display a second-order transition to synchronization, there are settings, such as when the node frequency and its degree are correlated, where this transition is first-order. This type of transition, called explosive synchronization, has been recently linked to pathological brain states. It appears thus of utmost importance to understand whether and how such type of transition can be suppressed in favor of a classical, smoother transition. Motivated by these considerations, in this paper we investigate synchronization in networks with frequency-degree correlation where oscillators are coupled via hubs-biased Laplacian operators, rather than the classical Laplacian. This corresponds to consider that the diffusive process underlying oscillator interactions is biased either towards or against high-degree nodes. We study analytically the effect of the hubs-biased Laplacians on synchronization in star-like networks and then numerically extend our study to other network configurations such as... Lastly, we show that in networks with hubs-biased Laplacian couplings the onset of explosive synchronization can be effectively controlled by small changes in the network topology.
% \end{abstract}

\begin{abstract}
    Nowadays, explosive synchronization is a well documented phenomenon
occurring in networks when the node frequency and its degree are correlated.
This first-order transition, which may coexists with classical synchronization,
has been recently causally linked to some pathological brain states
like epilepsy and fibromyalgia. It is then intriguing how most of
neuronal systems can operate in normal conditions avoiding explosive
synchronization. Here, we have discovered that synchronization in
networks where the oscillators are coupled via degree-biased Laplacian
operators, naturally controls the transition from explosive to standard
synchronization in neuronal-like systems. We prove analytically that explosive synchronization emerges when using this theoretical setting in star-like (neuronal) networks. As soon as this star-like network
is topologically converted to a network containing cycles, e.g., via
synaptic connections to other neurons, the explosive synchronization gives rise to
classical synchronization. This allows us to hypothesize that such
topological control of explosive synchronization could be a mechanism
for the brain to naturally work in normal, non-pathological, conditions.
\end{abstract}

\maketitle

\begin{quotation}
Here we report how explosive synchronization can be topologically
regulated in a neuron-like branched tree just by the cycles created
when it is glued (synaptically) to another neuron. The sufficient condition for this topological regulation is that the synchronization
process is controlled by a degree-biased diffusion. Therefore, there
is a transition between explosive synchronization in a branched acyclic
system to normal one once a cycle emerges in the system. This
transition represents a potential mechanism with which a neuronal
system can synchronize explosively individual neurons, and returning
to normal sync when the neuronal network is formed to avoid pathological
states like epilepsy or chronic pain. 
\end{quotation}

\section{Introduction}

Synchronization is one of the most common examples of collective behaviors
in complex systems \cite{Sync_review}. The transition from a disordered
state to the one in which all the nodes oscillates with the same phase
typically occurs in a gradual way, which is characteristic of second order transitions.
Therefore, the seminal discovery of ``explosive synchronization''
(ES) by Gomez-Gardeñes et al. \cite{ES_1} in 2011 showing that networked
Kuramoto oscillators may display a transition from a non-synchronized
state to a synchronized one in an abrupt/discontinuous way, marked
a tipping point in this field. According to the results in [\onlinecite{ES_1}]
there are two sufficient conditions for the onset of ES. First, the
network should display a scale-free (SF) topology and second, there
should be a positive correlation between the natural frequency of
an oscillator and its degree, i.e., its number of connections in the
network. Since the publication of this work there has been an explosion
of publications on ES (see \cite{ES_review} for a review), which
have been mainly focused on the conditions of its onset \cite{ES_2,zhang2013explosive,ES_4,ES_5, complexsync}.

Apart from being mathematically and aesthetically elegant, ES has
also being falsified on experiments of electronic circuits in a star
configuration, where it was demonstrated the existence of a first-order
transition towards synchronization of the phases of the networked
units \cite{ES_Exp}. The only remaining ingredient for ES to become
a fundamental pillar of complex systems would be its observation in
real-world complex networks. An important setting for searching ES
in the real scenarios are neuronal systems. The synchronization of
neurons plays a major role in reaching and keeping normal brain functions,
while its failure is related to certain brain disorders, such as Parkinsons
disease and epilepsy. In 2016, Kim et al. reported that ES conditions
are not present in the resting state of normal human brains \cite{Normal_brain}.
On the contrary, ES has emerged as an important factor of certain
pathological states of human brains. For instance, it has been demonstrated
that the fibromyalgic brain displays characteristics of ES conditions,
and that ES significantly correlates with chronic pain intensity. In
fact, Kaplan et al. \cite{Pain_2} have proposed potential palliative
intervenctions against pain which are based on the suppression of ES.
In another work, Wang et al. \cite{Epilepsy} have demonstrated that
ES-like transition occurs around the clinically defined onset of a
seizure episode from a patient with epilepsy. Accordingly, ES can be considered
as a possible mechanism for the recurrence of epileptic seizures.
Another report accounts for the possible involvement of ES in the
unexpected cognitive lucidity and communication in patients with severe
dementias, which is known as ``paradoxical lucidity'' and generally
occurs around the time of death of patients \cite{Paradoxical_lucidity}.
All of these findings point to the hypothesis that ES may be an important
factor, not in the functioning of normal brain, but in some of its
pathological states.

Due to these potential involvements of ES on pathological brain states,
some authors have investigated ways to suppress this kind of first
order transition in the network synchronization. Some of the strategies
proposed are, for instance, changing a small fraction of oscillators
to have negative couplings \cite{Suppresion_1, CompetitiveSuppression, jalan2019inhibition}, or by targeting (removing)
the hubs of the network \cite{Pain_2}. However, what it is intriguing
is how a normal brain works by naturally avoiding ES. That is, if,
as Danziger et al. \cite{Coexistence} have shown, the explosive phase
coexists with the standard phase of the Kuramoto oscillators, how
is it possible that all brains are not epileptic or continuously feeling
pain? In other words, if a neuronal system needs to synchronize for
doing its normal functioning, there should be a ``natural'' mechanism
which allows synchronization without ES. Here we investigate this question by adopting a minimalist description of the dynamics of the units, which are modelled using phase (Kuramoto-like) oscillators. Further works can adapt our theoretical framewok to more realistic synchronization scenarios in neuronal systems.

First, we consider the hypothesis that a neuronal tree, such as a
dendritic \cite{dendrite_1,dendrite_2} or a starburst \cite{startburs_1,startburst_2}
neuron can synchronize explosively, such that a perturbation rapidly
propagates across the whole neuron. However, the microscopic mechanism
governing the synchronization process is biased by the degree of the
nodes. That is, instead of surgically removing or adding hubs to the
network \cite{Pain_2}, our model naturally generalizes to a Laplacian
operator that avoids or preferentially uses the nodes of higher degrees
among all nearest neighbors of a given node \cite{hubs_Laplacian_1,hubs,hubs_Laplacian_3,Advection}.
Our main discovery here is that although such individual neuron can
synchronize explosively, when it produces a synaptic connection with
another neuron, such that cycles are formed, the ES is naturally suppressed.
We show that the inhibition of ES is more significant when the length
of the cycles formed during the synapsis is smaller. Therefore, normal
and abnormal neuronal states are induced just from the topological
nature of the synaptic connections. 

\section{Model of coupled phase oscillators}

Let $G=(V,E)$ be an undirected graph without self-loops or multiple edges, with a number of vertex $|V|=n$ and adjacency matrix $A$. For any node $i$ of the graph $G$, the degree of the node is defined as $k_i = \sum\limits_{i=1}^n A_{ij}$. Let $k$ be a column vector of the node degrees and let $K=\mathrm{diag}(k)$ be the diagonal matrix of the node degrees. We call hubs to the nodes with high degree and leaves to the nodes with $k_i=1$.

Let $\theta_i \in [0,2\pi]$ indicate the phase variable associated with node $i$ of the graph $G$. Then, the Kuramoto model on the graph $G$ can be expressed as:
\begin{equation} \label{eq:Kuramoto}
    \Dot{\theta_i} = \omega_i + \sigma \sum\limits_{j=1}^n L_{ij} \sin(\theta_i - \theta_j),
\end{equation}
where $i=1,\ldots,n$; $\sigma$ represents the coupling strength, and $\omega_i$ the natural frequency of the node $i$. The $n\times n$ matrix $L$ is the operator modeling how interactions take place in the graph. In the standard Kuramoto model this is the classical Laplacian, i.e., $L=L_c=K-A$.

In this paper, we generalize the Kuramoto model by considering, along with the classical Laplacian, two other Laplacian operators, namely the degree-biased ones \cite{hubs_Laplacian_1,hubs}. These are defined as $L_\alpha = \mathrm{diag}(K^{-\alpha} A K^\alpha)-K^{-\alpha} A K^\alpha$, where $\alpha \in \{-1,1\}$ is the biasing parameter, such that $L_1=L_a$ is the hubs-attracting Laplacian and $L_{-1}=L_r$ is the hubs-repelling Laplacian. 

Aiming at investigating the phenomenon of ES, we will study the Kuramoto model with frequency-degree correlation as in [\onlinecite{ES_1}], by setting $\omega_i = k_i$. Different correlations may not produce ES as shown in [\onlinecite{local}]. To monitor the system behavior, we will use a series of measures and order parameters that are summarized below.

The first of these measures is the Kuramoto order parameter $r$, defined as
\begin{equation}
 r(t) = \left\lvert \frac{1}{n} \sum_{j=1}^n e^{\mathrm{i} \theta_j(t)} \right\lvert
\end{equation}
with $\mathrm{i}=\sqrt{-1}$. The Kuramoto order parameter $r$ takes values in $[0,1]$, with $r\simeq 1$ indicating a high level of coherence between the phases of the oscillators, while  $r\simeq 0$ represents an uncorrelated behavior of the oscillators. For numerical purposes, we will mainly calculate $R=\langle r(t) \rangle_T$, for $T$ being a sufficiently large window of time, after the transient dynamics, where the behavior of the parameter is averaged.

In addition to that, as in [\onlinecite{Kuramoto_partial}], we also consider a measure of phase synchronization for each pair of nodes $i$ and $j$, by defining the pairwise Kuramoto order parameter:
\begin{equation}
 r_{ij} = \langle|e^{\mathrm{i} (\theta_i(t) - \theta_j(t))}|\rangle_T
\end{equation}
From this pairwise Kuramoto order parameter, we can define two indicators that are specific for the star graph and allow to measure the level of synchronization between the leaves of the graph, and the coherence between the leaves and the hub. In more detail, let the hub be node 1, then the coherence among the leaves is measured by: 
\begin{equation}
R_L = \frac{2}{(n-1)(n-2)} \sum\limits_{i=2}^{n} \sum\limits_{j>i} r_{ij}
\end{equation}
and that between the leaves and the hub is given by:
\begin{equation}
R_H = \frac{1}{n-1} \sum\limits_{i=2}^{n} r_{1i}
\end{equation}

Another information that we will monitor is the limit average value of the instantaneous frequency of the oscillators that we will denote by $\Omega_c$, $\Omega_a$ and $\Omega_r$ for the model with coupling described by the classical, hubs-attracting and hubs-repelling Laplacians respectively. This is defined as follows:
\begin{equation}
\Omega=\langle \lim\limits_{t\rightarrow +\infty} \dfrac{\theta_i(t+\delta t)-\theta_i(t)}{\delta t} \rangle_i
\end{equation}
where $\delta t$ is the integration step size and $\langle \cdot \rangle_i$ indicates averaging over the ensemble of network nodes.
    
Finally, we introduce a measure of the distance between clusters of nodes. In more detail, let $C_1,C_2,\ldots,C_m$ represent a set of clusters of nodes, namely a partition of the network, such that $C_i \cap C_j=\emptyset$ for any $i$ and $j$, and $\bigcup\limits_{i=1}^m C_i=V$. Then we can build an $m\times m$ matrix, denoted as $\Phi$, with entries $\Phi_{ij}= |m_i - m_j|$, where $m_i = \frac{1}{|C_i|(|C_i|-1)}\sum\limits_{j,l\in C_i} e^{\mathrm{i}(\theta_j(T)-\theta_l(T))}$  is the mean of the phases $\theta$ in the cluster $C_i$ and $T$ is as before. Notice that the matrix $\Phi$ is symmetric. From this matrix we will calculate the maximum eigenvalue (denoted as $\lambda_c$, $\lambda_a$ and $\lambda_r$ for the classical, hubs-attracting and hubs-repelling case, respectively), due to the fact that this eigenvalue is a better indicator than the mean or the second moment for measuring the size of the distances between the clusters \cite{Estrada07Chem}. In particular, in this paper we focus on clusters as formed by nodes with equal degree, that is $C_i = \{v \in V : k_v=i\}$ for $i=1,\ldots,m$, with $m$ equal to the maximum degree of the graph.

\section{Analytical results}

First of all, we observe that, as well as the classical Laplacian, degree-biased Laplacians are zero-row sum.  Hence, they also have a zero eigenvalue with an associated eigenvector $\Vec{1}$ for a connected graph as the ones used here. For this reason, we expect that long-time synchronized dynamics of Eqs.~\eqref{eq:Kuramoto} is supported by both the diffusion processes associated to hubs-repelling and hubs-attracting Laplacians, that is, for $L=\{L_r,L_a\}$. However, it is important to understand when such synchronous behavior can be observed, i.e., for which values of the coupling strength, and the type of transition with respect to this parameter that is found. More specifically, since an explosive transition to synchronization has been described for the classical Laplacian case~\cite{ES_1}, it is interesting to study analytically the critical point for synchronization. In doing so, we highlight how the result in~\cite{ES_1} can be retrieved as a special case of our approach, indicating the general nature of the current method. Here we consider the extreme case of a star graph, where nodes are split into nodes having degree one (the leaves) and a single node having degree $n-1$ and so acting as a hub (the center of the star). This kind of graph has been commonly used to study the simplest example of 
degree heterogeneity in many papers, including studies focusing on synchronization \cite{ES_1}. At the same time, it should be remarked that, because the degree-biased Laplacians are conceptually different operators than the classical one, in general one expects to find different characteristics, for instance in the configurations of the clusters arising in the network of oscillators or in the behavior with respect to the coupling strength. In particular, we expect to find different limit frequencies and distances between clusters when diverse Laplacians are used in Eqs.~\eqref{eq:Kuramoto}. These aspects will be investigated numerically in the next section, which also considers other network topologies.

The next result states our main finding for the star configuration. 
\begin{teor} \label{teor:analytic}
Let $G$ be a star graph with $n$ nodes and central node labelled as $1$. Let $\theta_j$ be the phase of node $j$ evolving according to the Kuramoto model \eqref{eq:Kuramoto} for the degree-biased Laplacians. Let $\Omega$ be the final angular speed that the nodes reach when the system synchronizes and consider the moving frame $\phi_j(t) = \theta_j(t) - \Omega t$. Then, for large enough coupling strength and time so the system is synchronized, we can obtain the phases of the nodes in this rotating frame as:
\begin{widetext}
\begin{equation} \label{eqs:star} \begin{split}
     \sin (\phi_{1}) & = \frac{\omega_{1}-\Omega}{\sigma\left(k_{1}+1\right) r}\left(\frac{k_{1}}{k_{j}}\right)^{\alpha}; \\
     \cos (\phi_{j}) & = \frac{\left(\omega_{j}-\Omega\right) \sin (\phi_{1}) \pm \sqrt{\left(1-\sin ^{2} (\phi_{1})\right)\left(\sigma^{2}\left(\frac{k_{1}}{k_{j}}\right)^{2 \alpha}-\left(\omega_{j}-\Omega\right)^{2}\right)}}{\sigma\left(\frac{k_{1}}{k_{j}}\right)^{\alpha}}.
\end{split}\end{equation}
\end{widetext}
for $j=2,\ldots,n$. Moreover, the critical coupling constant, namely the minimum value of the parameter $\sigma$ for which the system synchronizes, is given by:
\begin{equation} \label{critical}
    \sigma_c = |\omega_j - \Omega|\left(\frac{k_1}{k_j}\right)^{-\alpha}.
\end{equation}
 \end{teor}
 
\begin{proof}
We start rewriting the Kuramoto equations \eqref{eq:Kuramoto} in the new moving frame, splitting them in one equation for the hub and one for the leaves as follows:
\begin{gather}
     \dot{\phi_1} = \omega_1 - \Omega + \sigma \sum_{j=2}^{n} \left(\frac{k_j}{k_1}\right)^\alpha \sin (\phi_j - \phi_1); \label{eq:hub} \\
     \dot{\phi_j} = \omega_j - \Omega + \sigma \left(\frac{k_1}{k_j}\right)^\alpha\sin(\phi_1 - \phi_j). \label{eq:leaf}
\end{gather}
 
Note that, since $r e^{\mathrm{i}\Omega t} = \frac{1}{n} \sum_{j=1}^{n} e^{\mathrm{i}\theta_j}$, in the new frame the Kuramoto order parameter $r$ can be written as:

\begin{equation*}
r = \frac{1}{n} \sum_{j=1}^{n} e^{\mathrm{i}(\theta_j -\Omega t)} = \frac{1}{n} \sum_{j=1}^{n} e^{\mathrm{i}\phi_j} = \frac{1}{n} \sum_{j=1}^{n} \cos(\phi_j).
\end{equation*}

Multiplying this equation by $e^{-\mathrm{i} \phi_1}$ and using it in Eq.~\eqref{eq:hub}, we obtain:

  \begin{equation*}
     \dot{\phi_1}= \omega_1 - \Omega - \sigma \left(\frac{k_j}{k_1}\right)^\alpha (k_1 +1) r \sin(\phi_1).
 \end{equation*}
 
Since the final angular speed is $\dot{\theta_h} = \Omega$ for $h=1,\ldots,n$, setting $\dot{\phi_1}=0$ leads to the first of Eqs.~\eqref{eqs:star}.
 
Then, expanding the sine term in Eq.~\eqref{eq:leaf} and using the identity

\begin{eqnarray*}
     \sigma^2 \left(\frac{k_j}{k_1}\right)^{2\alpha} - \left(\Omega - \omega_j\right)^2 - \sigma^2 \sin^2(\phi_1)  + \sin^2(\phi_1)\left(\Omega - \omega_j\right)^2 = \\ 
     \left(1-\sin^2(\phi_1)\right)\left(\sigma^2\left(\frac{k_j}{k_1}\right)^{2\alpha}-\left(\omega_j-\Omega\right)^2\right),
\end{eqnarray*}
we can solve for $\cos(\phi_j)$ to obtain the second of Eqs.~\eqref{eqs:star}.
 
Since this computation is meaningful only if the terms inside the square root are non-negative, we get the condition for the critical coupling for synchronization:

  \begin{equation*}
     \sigma\left(\frac{k_j}{k_1}\right)^{\alpha}\geq\left|\omega_j-\Omega\right|;
 \end{equation*}
from which Eq.~\eqref{critical} immediately follows.
\end{proof}

As in [\onlinecite{ES_1}], we can use Eqs.~\eqref{eqs:star} and \eqref{critical} to calculate the corresponding value $r_c$ for the Kuramoto order parameter. Furthermore, for any $\sigma \geq \sigma_c$ we can compute the Kuramoto order parameter using the same equations, but there are two different values depending on the choice of sign at the cosine equation. This means that for $\sigma \ge \sigma_c$ there are two different locked-phase solutions that possibly converge to different values of $r$.

Notice that, as $k_1>>k_j$ when $n$ is sufficiently large, Eq.~\eqref{critical} provides an analytical arguments for the results found in [\onlinecite{hubs}], where the Kuramoto model without frequency-degree correlation in presence of the three different Laplacian operators is studied. If we assume that the difference between the frequency of the leaves $\omega_j$ and the final angular speed $\Omega$ is the same for all operators, then it follows that the hubs-attracting Laplacian ($\alpha=1$) needs a lower coupling to synchronize than the classical Laplacian ($\alpha=0$), while the hubs-repelling Laplacian ($\alpha=-1$) needs a higher one, as indeed found in [\onlinecite{hubs}].

Moreover, here we also observe an amplification of the finding illustrated in [\onlinecite{f-d}], where it is shown that synchronization is reached earlier in nodes with large degrees, while lower degree nodes synchronize later. This may be counter-intuitive because of the names, hubs-attracting and hubs-repelling Laplacians. However, we note that, in the Kuramoto model, the term $L_{ij}$ represents the influence from node $j$ to node $i$, such that the degree-biased Laplacians effectively acts on the Kuramoto model via their adjoint operators, as defined in [\onlinecite{Advection}]. Notably, in [\onlinecite{Advection}] it is shown that the adjoint operator reverses the directions in the network, so the $ij$-th term now effectively represents the influence of node $j$ into node $i$. This causes that the flow towards low degree nodes is bigger when using the adjoint hubs-attracting Laplacian and, conversely, the adjoint hubs-repelling Laplacian shows a bigger flow towards high degree nodes.

Theorem \ref{teor:analytic} relies on the fact that all the nodes reach the same final angular speed once the system is synchronized. Thus, it is interesting to calculate the value of $\Omega$ analytically.

\begin{teor} \label{teor:finalspeed}
Under the same assumptions of Theorem \ref{teor:analytic}, the final angular speed $\Omega$ of the synchronized system is given by:
\begin{equation}\label{eq:finalspeed}
    \Omega=\frac{(n-1)^{\alpha+1} + (n-1)^{1-\alpha}}{(n-1)^\alpha+(n-1)^{1-\alpha}}
\end{equation}
\end{teor}
\begin{proof}
As in the proof of Theorem \ref{teor:analytic}, we use a rotating frame, hence working with equations \ref{eq:hub} and \ref{eq:leaf} for the hub and the leaves, respectively. Let us assume that the system is synchronized, then we have that $\dot{\phi_i}=0$ for all $i$.

On the other hand, set any $j \in \{2,\dots,n\}$ and calculate the sum:
\begin{widetext}
\begin{equation*}
    \left(\frac{k_1}{k_j}\right)^\alpha \dot{\phi_1} + \sum_{i=2}^n \left(\frac{k_i}{k_1}\right)^\alpha \dot{\phi_i} = \left(\frac{k_1}{k_j}\right)^\alpha (\omega_1 - \Omega) + n \left(\frac{k_j}{k_1}\right)^\alpha (\omega_j - \Omega) + \sigma \sum_{i=2}^n  \left(\sin(\phi_1 - \phi_i) + \sin(\phi_i - \phi_1) \right).
\end{equation*}
\end{widetext}

As $\dot{\phi}_i=0$ $\forall i$, the left side of the sum is equal to zero as well, while we can use the fact that the sine is an odd function to simplify the right side. After rearranging the terms we obtain:
\begin{equation*}
    \left(\left(\frac{k_1}{k_j}\right)^\alpha + n \left(\frac{k_j}{k_1}\right)^\alpha \right) \Omega= \left(\frac{k_1}{k_j}\right)^\alpha \omega_1 + n \left(\frac{k_j}{k_1}\right)^\alpha \omega_j.
\end{equation*}

Finally, using that $k_1=\omega_1=n-1$ and $k_j=\omega_j=1$, we get the formula \ref{eq:finalspeed}.
\end{proof}

The solutions given in Theorem \ref{teor:analytic} only exist if $\sigma \geq \sigma_c$ and they represent coherent states. Thus, as it is shown in [\onlinecite{ES_1}], the value $\sigma_c$ can be used to predict the drop of coherence on the backward continuation. Nevertheless, it is not sufficient to show that a first order transition occurs. For that purpose, we need to calculate the value $\sigma_s$ for which the system suddenly passes from incoherent to coherent steady states. Moreover, if the value $\sigma_s > \sigma_c$, hysteresis will appear, which is an indicator of first order transitions.

Calculating the value of $\sigma_s$ is more difficult than doing it for $\sigma_c$, even for the classical Laplacian. For doing so, we will follow the steps of articles \cite{VlasovForward} and \cite{Jiang}. To work out the calculations some assumptions are needed. Remarkably, either the Watanabe-Strogatz \cite{Watanabe-Strogatz} or Ott-Antonsen \cite{Ott-Antonsen} ansatz can be used to arrive to the same result.

\begin{teor} \label{teor:odes}
Let $G$ be a star graph with $n$ nodes and central node labelled as $1$. Let $\theta_j$ be the phase of node $j$ evolving according to the Kuramoto model \eqref{eq:Kuramoto} for the degree-biased Laplacians. Let $r e^{i\psi} = \sum\limits_{j=1}^n e^{i\theta_j}$. Then, for long enough time so the system achieves a steady state, $(r, \psi)$ are solutions to the system of ODEs:
\begin{eqnarray} \label{eqs:Rpsi}
    \dot{r}=&& \frac{\sigma}{2}\left((n-1)^{\alpha}-(n-1)^{\alpha} r^2\right)\cos{\psi}; \\
    \dot{\psi}=&& -\frac{\sigma}{2r} \left( (2(n-1)^{1-\alpha} + (n-1)^{\alpha})r^2+(n-1)^{\alpha}\right) \sin{\psi} + n - 2. \nonumber
 \end{eqnarray}

Moreover, calculating the equilibrium points of the previous system of equations, we can obtain the values of $\sigma_c$ and $\sigma_s$ for the degree-biased Laplacians:
\begin{eqnarray}
    \sigma_c&&=\frac{n-2}{(n-1)^{1-\alpha} + (n-1)^\alpha}, \nonumber \\
    \sigma_s&&=\frac{n-2}{\sqrt{(n-1)^{2\alpha} + 2n -2}}. \label{sigmas}
\end{eqnarray}
\end{teor}
\begin{proof}
We will follow the proofs in [\onlinecite{Jiang}] and [\onlinecite{XuForward}].

Define $\phi_j=\theta_j-\theta_1$ for $j=2,...,n$. Then it satisfies the equation:
\begin{equation*}
    \dot{\phi_j}=1-n-\sigma (n-1)^{-\alpha} \sin{\phi_j} - \sigma (n-1)^{\alpha} \sum_{l=2}^n \sin{\phi_l} 
\end{equation*}

Let $n$ be large enough so $\rho(\phi,t)$, the probability density function that gives us the fraction of oscillators with phase $\phi$ at time $t$, is well defined. Consider $D= \int_0^{2\pi} \rho \sin{\phi} d\phi$. Then, all oscillators have angular velocity $v= 1-n - \sigma \sin{\phi} - \sigma n D$, and we can obtain that $\rho$ satisfies the equation $\dfrac{\partial \rho}{\partial t} + \dfrac{\partial \rho v}{\partial \phi} =0$.

Now, we expand $\rho$ in Fourier series and use the Ott-Antonsen ansatz, so the $l$-th Fourier coefficient is equal to the $l$-th power of first one. Let $z$ be this first Fourier coefficient of $\rho$, then we can use the equation obtained for this probability density function to obtain the equation for $z$:
\begin{eqnarray*}
        \dot{z}=\frac{\sigma}{2} - \sigma\frac{(n-1)^(\alpha+1+(n-1)^{-\alpha})}{2}z^2 \\ - i (1-n) z + \sigma \frac{(n-1)^{\alpha+1}}{2} |z|^2
\end{eqnarray*}
where we have used that $D= Im(z)$.

Moreover, as $z$ is the first coefficient of the Fourier expansion of $\rho$, we have: $|z|= \left| \int \rho r^{i\phi} d\phi \right| = r$. Thus, we can write $z=r e^{i\psi}$ and substitute it in the equation for $z$. Multiplying by $e^{-i\psi}$, we can decouple the equations in \eqref{eqs:Rpsi} by taking the real and imaginary parts, respectively.

For obtaining the values of $\sigma$, we have to find the equilibrium points of \eqref{eqs:Rpsi}.

Solving the equation $\dot{r}=0$, we obtain that either $r=1$ or $\psi=\frac{\pi}{2}$. Substituting both values into the equation $\dot{\psi}=0$ we obtain:
\begin{widetext}
\begin{eqnarray}
    \psi=\frac{\pi}{2}, r=\frac{2(n-2) \pm \sqrt{4 (n-2)^2- 4 \sigma^2 (n-1)^\alpha (2(n-1)^{1-\alpha}+(n-1)^\alpha)}}{2\sigma (2(n-1)^{1-\alpha}+(n-1)^\alpha)}; \label{incoherent}\\
    r=1, \psi=\arcsin{\frac{n-2}{\sigma ((n-1)^{1-\alpha} + (n-1)^\alpha)}},\arcsin{\frac{-(n-2)}{\sigma ((n-1)^{1-\alpha} + (n-1)^\alpha)}}+\pi. \label{coherent}
\end{eqnarray}
\end{widetext}

Equation \eqref{incoherent} exists when the term inside the square root is positive, which gives us that these solutions exists for \\ $\sigma \leq \sigma_s=\dfrac{n-2}{\sqrt{(n-1)^{2\alpha} + 2 n - 2}}$.

On the other hand, as the domain of the $\arcsin$ function is the interval $[-1,1]$, Eq.~\eqref{coherent} is valid for \\ $\sigma \geq \sigma_c = \dfrac{n-2}{((n-1)^{1-\alpha} + (n-1)^\alpha)}$
\end{proof}

\begin{obs}
 Notice that Theorems \ref{teor:analytic} and \ref{teor:odes} provide two expressions for $\sigma_c$. Using the expression for $\Omega$ given by Theorem \ref{teor:finalspeed} it immediately follows that the two expressions are equal to each other:
 \begin{eqnarray*}
     |\omega_l -\Omega|\left(\frac{k_1}{k_l}\right)^{-\alpha} = \\ 
     (n-1)^{-\alpha} \left(\frac{(n-1)^{\alpha+1} + (n-1)^{1-\alpha}}{(n-1)^\alpha+(n-1)^{1-\alpha}} -1 \right) = \\
     \dfrac{n-2}{((n-1)^{1-\alpha} + (n-1)^\alpha)};
 \end{eqnarray*}
 
 On the one hand, Theorem \ref{teor:odes} allows to calculate this value without using the final angular speed, which a priori is not known, but it is needed to assume $n$ large.
 
 On the other hand, Theorem \ref{teor:analytic} does not need this assumption. This allows to capture the cluster formation in the model, so that the first order transition does not actually reach $r=1$ for $n$ small.
\end{obs}
\begin{obs}
We can study the hysteresis width of the system by looking at Eq.~\eqref{sigmas}. As mentioned  in [\onlinecite{Jiang}], for the classical Laplacian ($\alpha=0$), $\lim\limits_{n \to \infty} \sigma_c=1$ and $\lim\limits_{n \to \infty} \sigma_s=\infty$, so the hysteresis gap increases as the star has more leaves. For the degree-biased operators, the situation changes.

In the hubs-attracting case ($\alpha=1$) it is interesting to notice that the value of $\sigma_c$ coincides with the one for the classical Laplacian, so $\lim\limits_{n \to \infty} \sigma_c=1$. But $\lim\limits_{n \to \infty} \sigma_s=1$ as well, so as the star size increases, the hysteresis tends to disappear.

If we take $\alpha=-1$, namely in the hubs-repelling case, then $\lim\limits_{n \to \infty} \sigma_s=\infty$ as for the classical Laplacian, although their values are different. This means that the hysteresis gap also increases with the size of the star. Nevertheless, as $\lim\limits_{n \to \infty} \sigma_c=0$, this gaps is bigger than in the classical case and, in the limit, we arrive to a situation where there aren't transitions between coherent and incoherent states when varying $\sigma$.
\end{obs}

\section{Numerical results}

\subsection{Star graphs}

We start by studying the model of coupled phase oscillators in star graphs in order to validate the analytical results obtained in the previous section. This kind of graphs is interesting as it represents a simplified case of what we expect to find in heavily branched neurons. The nodes of the star graph, displaying only two possible values for their degrees, also allow for a clear analysis of the effects of the degree-biased Laplacians that we are studying. Finally, another reason to study these simple graphs is that ES also occurs when the model with the classical Laplacian operator is considered for these graphs \cite{ES_1}. In particular, in this section we analyze how synchronization is affected by the application of degree-biased Laplacians. We also study the steady-state of the system, obtaining the angular speed of the nodes as well as the distance between the hub and the leaves of the star graph for different number of nodes.

To calculate the critical coupling given by Eq.~\eqref{critical}, we first need to derive the final angular speed $\Omega$ reached by the system. In Fig.~\ref{fig:gamma1speed} we illustrate the analytical results obtained by means of Theorem 3.2 and compare them with the numerical simulations carried out on graphs with a number of nodes varying
from $n = 10$ to $n = 40$.

We observe that the fastest speed is obtained for the hubs-attracting model, followed by the classical and then the hubs-repelling case. In particular, we note the remarkable fact that in the hubs-attracting model the final frequency grows linearly with the size of the star. The angular speed in the hubs-repelling case is the only one that decreases with the number of nodes of the graph.

The analysis of the final angular speed $\Omega$ reveals that the term $\omega_l - \Omega$ appearing in Eq.~\eqref{critical} increases for the hubs-attracting operator and decreases for the hubs-repelling one, while the term $\left(\frac{k_1}{k_j}\right)^{-\alpha}$ does the opposite. Consequently, it is not trivial to predict which Laplacian operator yields the lowest critical coupling.

From Eq.~\eqref{eqs:star} we derive that leaves and hub have different phases, such that a cluster of phase-synchronous nodes corresponding to the leaves emerges. To better characterize this interesting feature, we have computed the distance between the two clusters appearing in the star graph (the one formed by the leaves and the singleton one formed by only the hub), as shown in Fig.~\ref{fig:gamma1dist}. There is a similarity in the behaviour of the cluster distance and the angular speed of the system. As for the classical and hubs-attracting Laplacians, they both increase with the number of nodes, while for the hubs-repelling Laplacian they both decrease.

\begin{figure*}[t]
    \centering
    \subfigure[]{\includegraphics[width=0.31\linewidth]{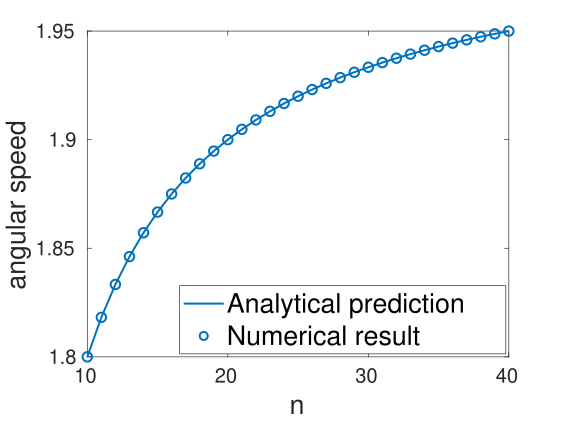}\label{fig:gamma1speedc}}
    \subfigure[]{\includegraphics[width=0.31\linewidth]{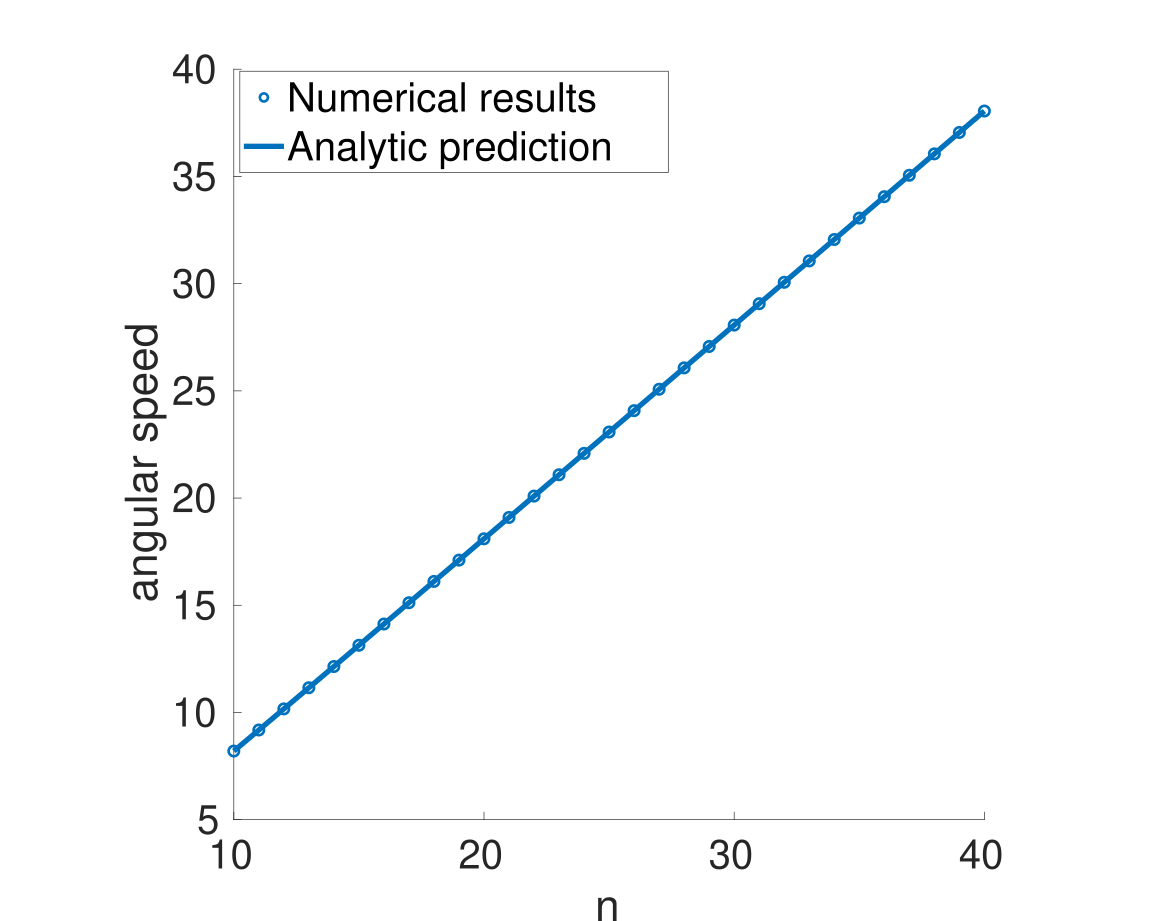}\label{fig:gamma1speeda}}
    \subfigure[]{\includegraphics[width=0.31\linewidth]{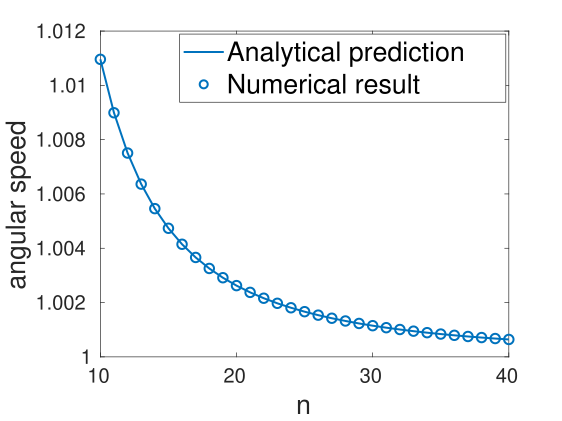}\label{fig:gamma1speedr}}
    \caption{Final angular speed $\Omega$ in a star graph for the classical (a), hubs-attracting (b) and hubs-repelling Laplacian (c).} \label{fig:gamma1speed}
\end{figure*}

\begin{figure*}[ht]
    \centering
    \subfigure[]{\includegraphics[width=0.32\linewidth]{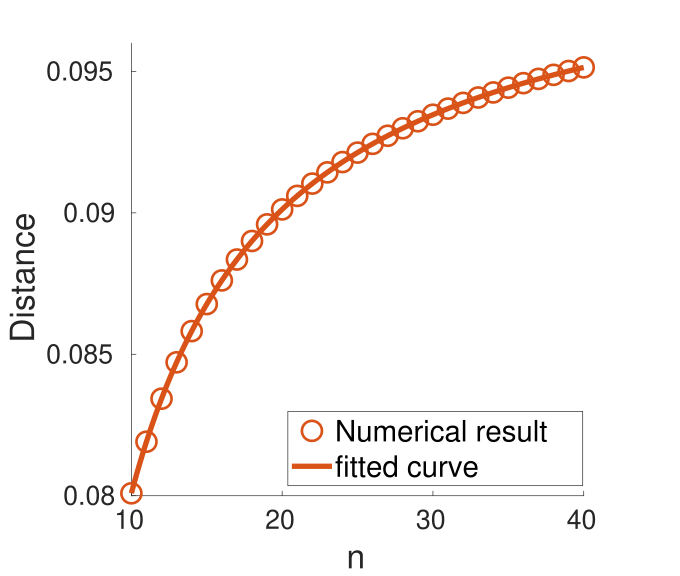}\label{fig:gamma1distclas}}
    \subfigure[]{\includegraphics[width=0.32\linewidth]{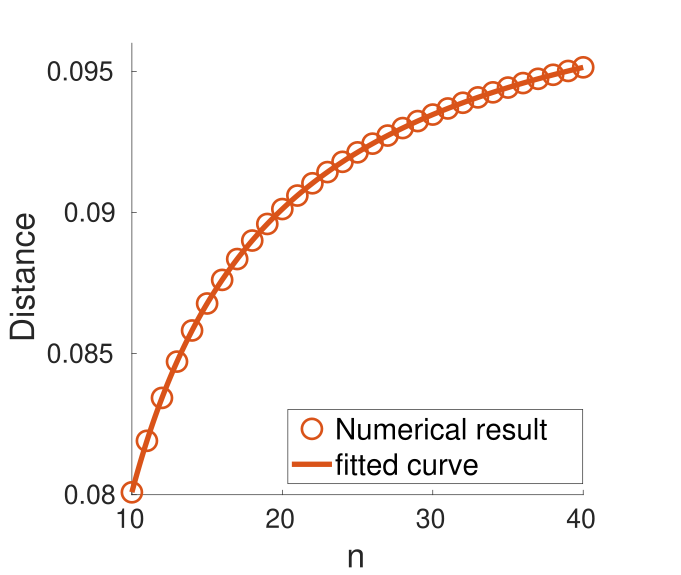}\label{fig:gamma1distattract}}
    \subfigure[]{\includegraphics[width=0.32\linewidth]{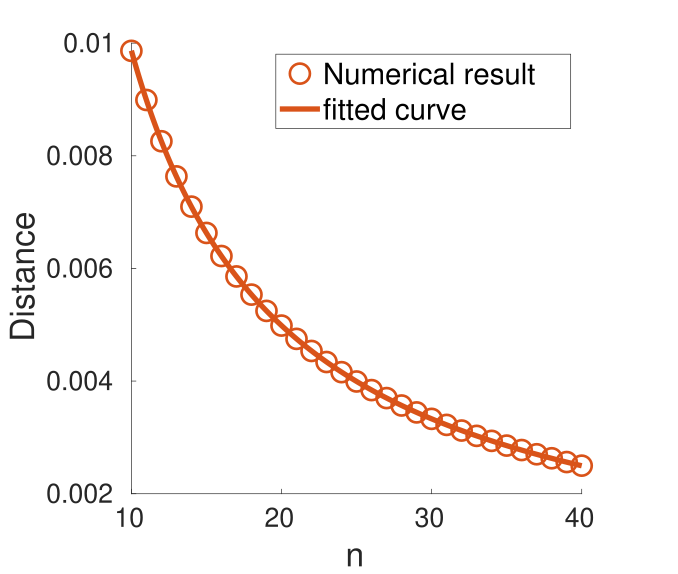}\label{fig:gamma1distrepell}}

    \caption{Cluster distance $\Phi$ in a star graph for the classical (a), hubs-attracting (b) and hubs-repelling Laplacian (c).} \label{fig:gamma1dist}
\end{figure*}

Next, we study the occurrence of ES. While for the classical Laplacian coupling it is known that the transition is first-order \cite{ES_1, f-d}, the other two Laplacians are analyzed here for the first time. To investigate the transition from disorder to synchronization, we monitor the Kuramoto order parameter $r$ as a function of the coupling parameter $\sigma$. The presence of the first-order transition is usually associated with an hysteresis in the value of the critical coupling when calculated by increasing or by decreasing $\sigma$. For this reason, we have computed the forward and backward continuation of $r$ increasing/decreasing $\sigma$ at steps of $0.01$. The numerical calculations have been carried out on a star graph with $n=10$, and so that $\omega_1=9$. The values of the natural frequencies of the leaves have been obtained by adding to the nominal value of $\omega_l=1$ a small perturbation, in the order of $o(10^{-3})$, accounting for some spatial disorder in the system.

The results for the three Laplacians are shown in Fig.~\ref{fig:rcontinuation}, together with the values $\sigma_c$ and $\sigma_l$ obtained by mean of Theorems \ref{teor:analytic} and \ref{teor:odes}.

Although we find an hysteresis in all the three cases, nevertheless, the new operators show interesting differences. First, both degree-biased Laplacians require a smaller coupling parameter $\sigma$ for the transition with respect to the classical Laplacian. This confirms the analytic prediction for $\sigma_l$. On the other hand, the backward desynchronization perfectly fits the analytic calculations, being equal for the classical and hubs-attracting Laplacian while it is very close to $0$ for the hubs-repelling operator.

It is important to note that, even for small $n$, the closeness between $\sigma_c$ and $\sigma_l$ in the hubs-attracting Laplacian case makes the transition to look smooth, but there is still a small hysteresis gap.

\begin{figure*}[ht]
    \centering
    \subfigure[]{\includegraphics[width=0.32\linewidth]{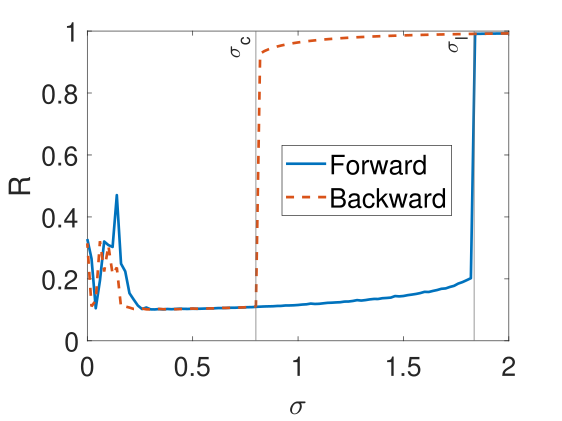} \label{fig:clas1}}
    \subfigure[]{\includegraphics[width=0.32\linewidth]{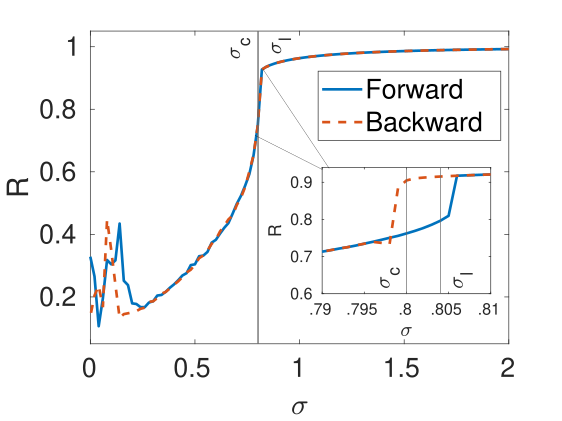} \label{fig:attract1}}
    \subfigure[]{\includegraphics[width=0.32\linewidth]{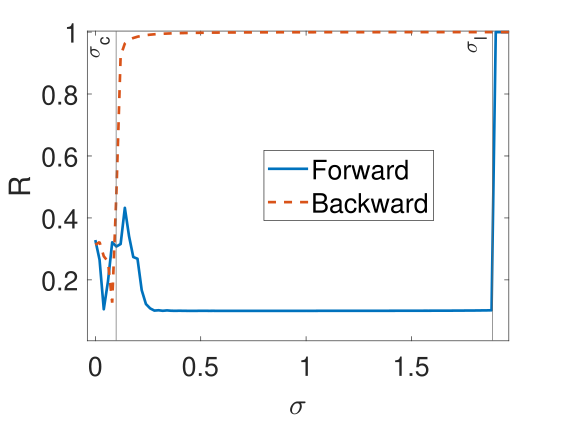} \label{fig:repell1}}
    \caption{Kuramoto order parameter $r$ vs. $\sigma$ for the classical (a), hubs-attracting (b) and hubs-repelling Laplacian (c).
    \label{fig:rcontinuation}}
\end{figure*}

If we compute the values of $R_L$ and $R_H$ to study separately the synchronization of the cluster of leaves and the synchronization of the hub with the leaves, we find that the values of $R_L$ are very close to $1$ in all cases, meaning that the cluster of leaves is quickly formed and holds for any $\sigma$. On the other hand the shape of $R_H$ is similar to the usual order parameter $R$, meaning that the hub synchronization is the one causing the first-order transition. The hubs-repelling case is the only one in which we can observe the influence of the ES in the parameter $R_L$, as the transition mainly affects the hub and this operator makes this central node to influence more to the leaves.

Another interesting difference in synchronization in star graphs with diverse Laplacian operators is observed when spatial disorder increases (Fig. \ref{noise}). In fact, when spatial disorder increases, the inhomogeneity among the frequencies increases while the correlation between frequency and degree at each leaf decreases. The behavior of the classical Laplacian is more stable with respect to the spatial disorder as, for every nonzero disorder, the hysteresis can be clearly seen by the difference between the forward and backward continuation. On the contrary, the hubs-repelling Laplacian does not desynchronize when the disorder is too small, so it is more difficult to appreciate the end of the hysteresis gap, and the hysteretic behaviour disappears when the size of the disorder increases above $10^{-2}$. The hubs-attracting Laplacian behaves similarly, but it is harder to appreciate due to the smallness of the hysteresis gap, so in Fig. \ref{attractnoise} only the cases where the disorder is nonzero are plotted and the axis limits are changed.

\begin{figure*}[ht]
    \centering
    \subfigure[]{\includegraphics[width=0.32\linewidth]{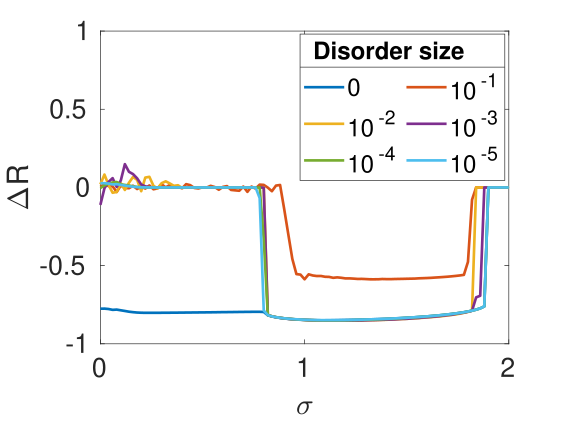} \label{clasnoise}} 
    \subfigure[]{\includegraphics[width=0.32\linewidth]{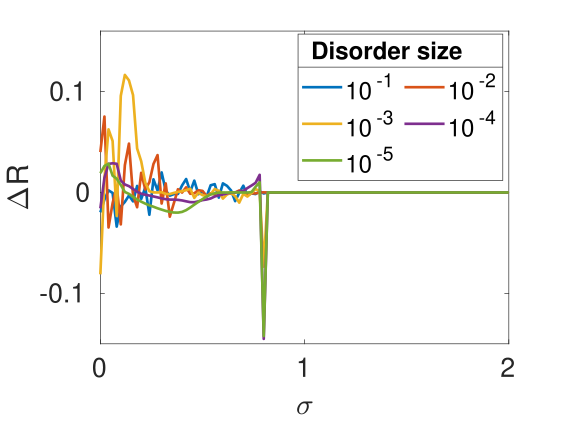}\label{attractnoise}}
    \subfigure[]{\includegraphics[width=0.32\linewidth]{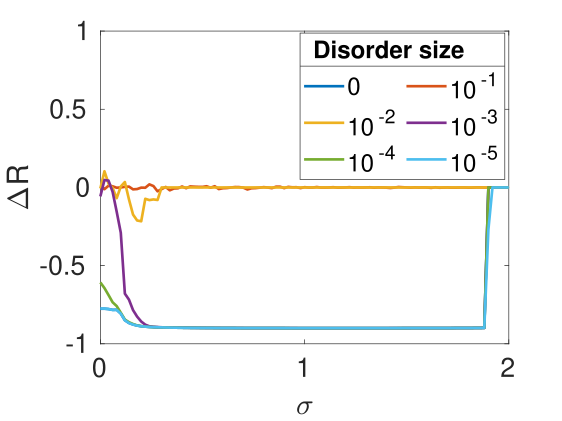}\label{repellnoise}}
    \caption{Difference between forward and backward continuations of the Kuramoto order parameter in a star graph for different values of the spatial disorder for the classical (a), hubs-attracting (b) and hubs-repelling Laplacian (c).}\label{noise} 
\end{figure*}

\subsubsection{Cayley irregular graphs}

First, we study system~\eqref{eq:Kuramoto} when the interaction network is modeled by an irregular Cayley graph of maximum degree $5$. Cayley graphs (sometimes named Bethe lattices) are trees are trees where there is a central node with degree k, then there is a first orbit of nodes of degree k-1, followed by a second orbit of nodes of degrees k-2, and so for up to arrive to an orbit of nodes with degree one.

From the simulations we observed that the synchronization in clusters is kept for the three operators, ordered from the highest to the lowest degree. We can use this defined order to study the distance between consecutive clusters. In Fig.~\ref{fig:CayleyDist} we can observe that this distance between clusters decays as the inverse of the coupling parameter $\sigma$. Then, if we make $\sigma \to \infty$, it is expected that the clusters also synchronize and all the nodes converge to the same point, such as when there is no frequency-degree correlation in the Kuramoto model.

\begin{figure}[ht]
    \centering
    \includegraphics[width =\linewidth]{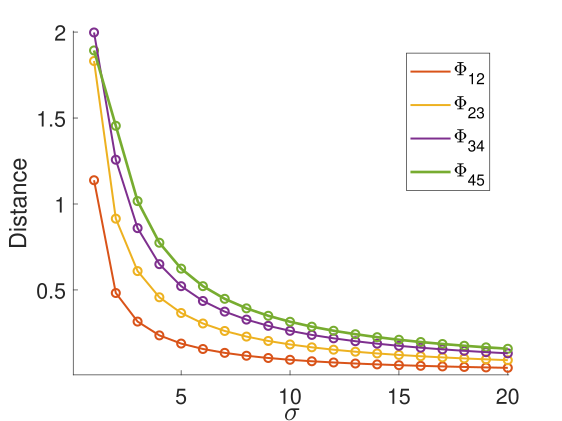} \label{fig:CayleyDist}
    \caption{Distances between consecutive clusters of the Cayley irregular graph with maximum degree $5$ for the hubs-attracting Laplacian.}
\end{figure}

Next we compute the maximum eigenvalue $\lambda$ of the cluster distance matrix, we obtained that $\lambda_r \le \lambda_a \le \lambda_c$, so clusters are closer if we use the hubs-repelling Laplacian, but any of the biased models decreases the distance between clusters with respect to the model with the classical Laplacian. 

On the other hand, we can observe from the simulations that the final angular speed of the nodes follows the same behaviour as for the star graphs $(\Omega_r < \Omega_c < \Omega_a)$. It is interesting to note that this angular speed is still invariant upon changing the coupling parameter $\sigma$ as long as it is over the value where the system synchronizes. 

In order to see if the explosive transition is kept in this kind of network, we compute its backward continuation too. As we can see in Fig. \ref{Cay}, there is hysteresis in all three cases. The behaviour is similar to that observed in the stars. The hysteresis gap is bigger in the hubs-repelling case compared to the one in the case of the classical Laplacian, while it is much smaller for the hubs-attracting Laplacian.

\begin{figure*}[ht]
    \centering
    \subfigure[]{\includegraphics[width=0.32\linewidth]{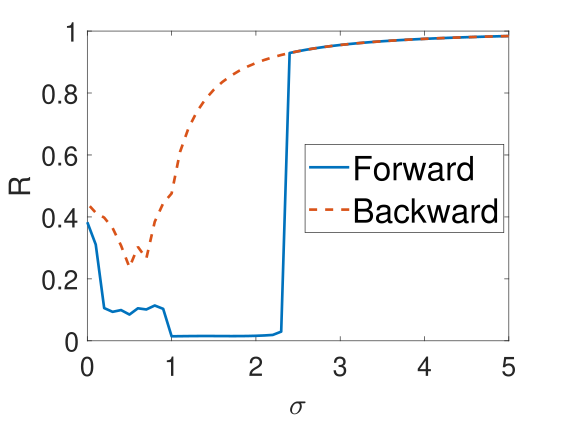}\label{CayHystc}} 
    \subfigure[]{\includegraphics[width=0.32\linewidth]{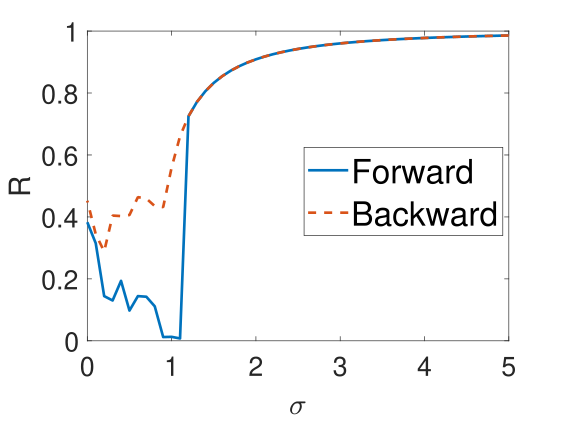}\label{CayHysta}}
    \subfigure[]{\includegraphics[width=0.32\linewidth]{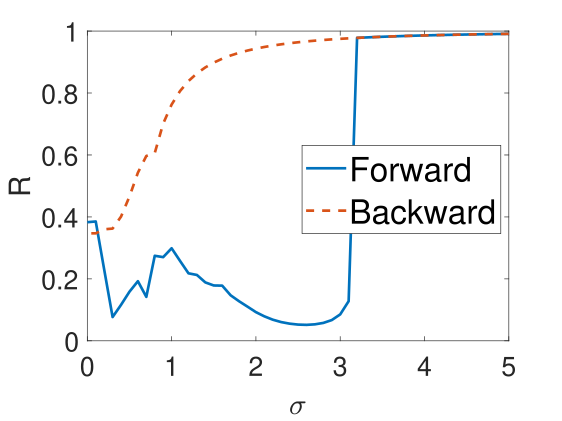}\label{CayHystr}}
    \caption{Forward and backward continuations for the Cayley irregular graph with maximum degree $5$ for the classical (a), hubs-attracting (b) and hubs-repelling Laplacian (c).} \label{Cay}
\end{figure*}

We also calculated the evolution of the Kuramoto order parameter along time for each of the clusters. They show that the time taken by each cluster to synchronize follows the same order than the critical coupling needed to synchronize the system. That is, the use of the hubs-attracting Laplacian reduces the synchronization time of all clusters compared to the use of the classical Laplacian, while using the hubs-repelling Laplacian increases it. We can also observe that, as the degree of the cluster decreases, there are more nodes in such cluster and thus it needs longer times to synchronize for any of the operators. 

\subsubsection{Scale-free random networks}

We now study the Kuramoto model with degree-frequency correlation in Barabási-Albert networks, which are scale-free random graphs, in which the degree follows a power-law distribution with exponent $-3$. In article [\onlinecite{ES_1}] the authors proved that hysteresis occurred in this type of networks as well. 

As we have only studied tree-like networks, we start by studying scale-free trees. That is, random graphs following the scheme proposed by Barabási and Albert but where the new nodes only attach to one existent node, so there are no cycles.

In Fig.~\ref{BATree} we observe a clear difference between the classical Laplacian and the degree-biased Laplacians. While the classical Laplacian keeps showing ES, the hubs-attracting Laplacians shows a smaller first order transition that does not reach $R=1$. The hubs-repelling Laplacian also shows hysteresis, but the discontinuous transition is not clearly seen. The difference now can be explained because, in the previous networks, the degrees decreased monotonically, so there was a preferred direction in the whole network (from the hub to the leaves or vice versa). But this scale-free trees are more asymmetric, allowing the hub to be connected to a low degree node which has another neighbour with higher degree. This results in lower values of the Kuramoto order parameter $R$.

\begin{figure*}[ht]
    \centering
    \subfigure[]{\includegraphics[width=0.32\linewidth]{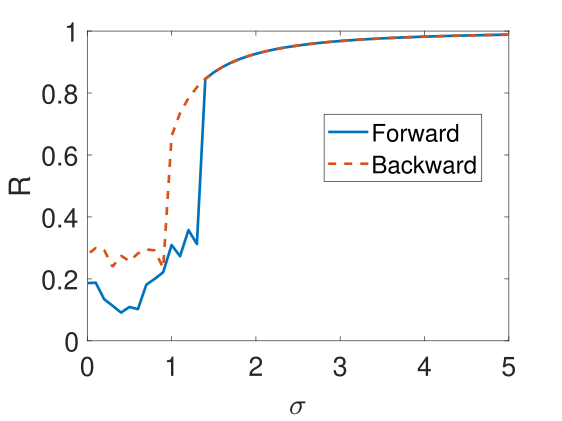}} 
    \subfigure[]{\includegraphics[width=0.32\linewidth]{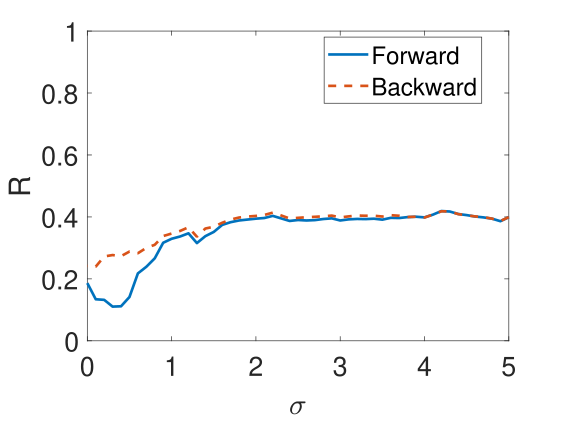}}
    \subfigure[]{\includegraphics[width=0.32\linewidth]{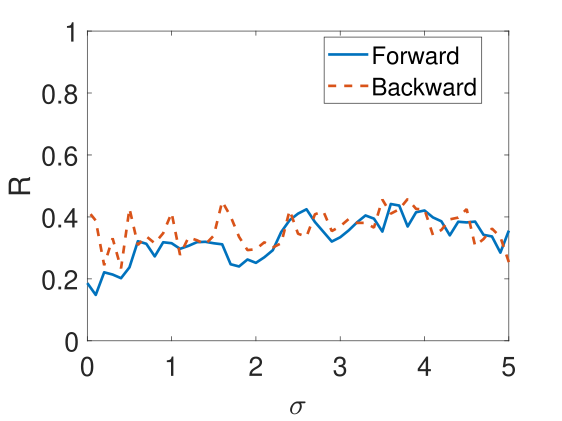}}
    \caption{ Kuramoto order parameter for the forward and backward continuations in a BA tree for the classical (a), hubs-attracting (b) and hubs-repelling Laplacian (c).} \label{BATree}
\end{figure*}

Although ES is reduced in scale-free trees, the hysteresis region can be clearly appreciated. But the asymmetries that contribute to this behaviour should be magnified by cycles. That is, general BA networks should reduce the hysteresis occurring when using any of the degree-biased Laplacians.

This is verified in Fig.~\ref{BANet}, where we see that the hysteresis is reduced for any of the three Laplacians. This shows that cycles generally reduce the hysteretic gap. Another interesting fact is that the values of the Kuramoto order parameter are higher in this case than in the tree network when using the hubs-attracting Laplacian.

In general, we can say that both the asymmetries in the degree of the nodes and cycles contribute to naturally inhibit the ES when we use the degree-biased Laplacians.

\begin{figure*}[ht]
    \centering
    \subfigure[]{\includegraphics[width=0.32\linewidth]{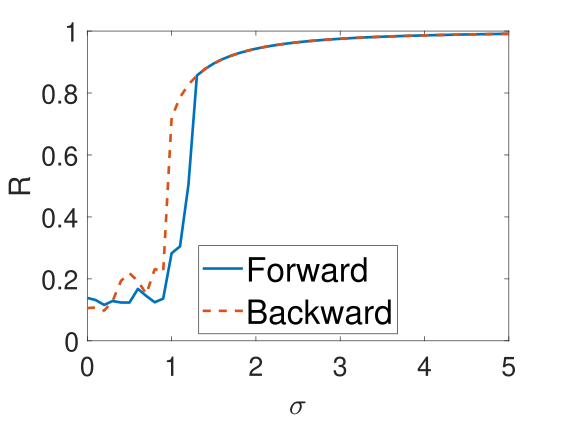}} 
    \subfigure[]{\includegraphics[width=0.32\linewidth]{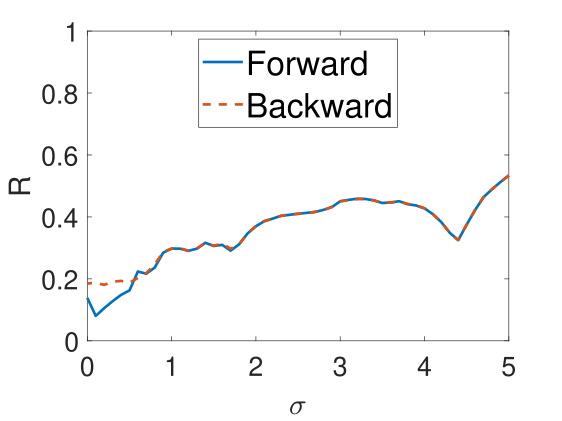}}
    \subfigure[]{\includegraphics[width=0.32\linewidth]{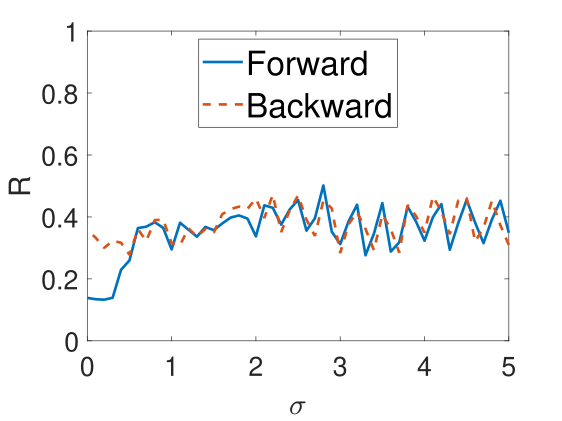}}
    \caption{ Kuramoto order parameter for the forward and backward continuations in a BA network for the classical (a), hubs-attracting (b) and hubs-repelling Laplacian (c).} \label{BANet}
\end{figure*}

\section{Suppression of the ES in neural networks}

It is known that explosive synchronization is causally related to
some pathological stets of the brain, such as epilepsy \cite{Epilepsy} and
chronic pain \cite{Pain_1}, as well as with the states related to consciousness
recovery \cite{Normal_brain}. It has been discovered that ES coexists with classical
synchronization in the Kuramoto model \cite{Coexistence}. Therefore, it is intriguing
to know how, when neuronal systems synchronize, they avoid ES to work
normally, i.e., avoiding such pathological states as the one mentioned
before. The previous section gives us some hints about the fact that
the degree-biased processes might be a good proxy for controlling
this situation. That is, we have shown that, depending on the
structure of the neural network (namely the asymmetry on the degrees
and the cycles in the network), the first order transition to the
synchronized state can be reduced. 

However, the degree-biased control alone could not be responsible
for the transition from ES to classical synchronization in neuronal
systems in a natural way. For instance, it is known that neurons are
of diverse shape \cite{startburs_1}, such that the topological structures of their
networks could be very different. This translates into a variety of
behaviours depending on the degree of asymmetry of such structures.
Those systems closer in their structure to Cayley irregular graphs
will most likely show ES than those whose shape looks like a scale-free
tree. 

Therefore, if we consider a single neuron with structure close to
Cayley irregular graphs it will synchronize explosively. This could
be advantageous to this neuron \cite{single_neuron1,single_neuron2} which will be
in a synchronous state in a very short period of time. However, the
problem emerges if, when this neuron makes synaptic connections with
many others forming a neuronal system, the synchronization of the
system as a whole is also explosive. From a topological perspective
this synapsis among groups of neurons creates, in the most general
and realistic way, cycles in the network structure. Therefore, we
will study what is the influence of the presence of cycles on the
maintenance of ES in a network when the process is controlled by the
degree-biased operators.

The first thing that we need to consider is the length of the cycles
that can be formed during the synapsis of two heavily branched neurons.
Due to the branched nature of these neurons it is very difficult that
the cycles formed among them are of small length, e.g., triangles.
It is more likely that large cycles emerge from the synaptic connections among
three of more neurons due to the geometric constraints produced by
the neuronal branching.

\begin{figure}[h]
\includegraphics[width=0.85\linewidth]{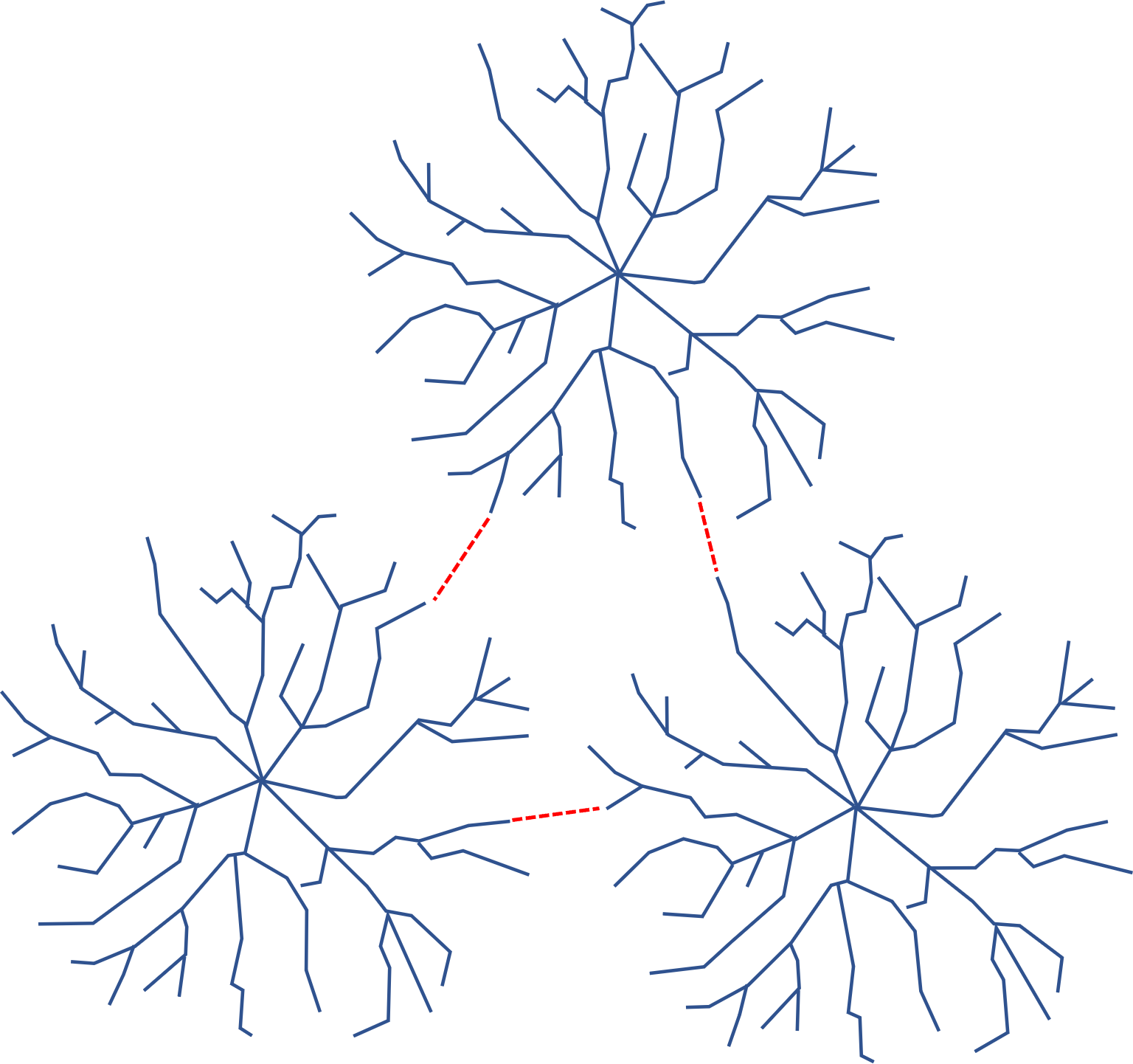}
\caption{Schematic illustration of the potential formation of cycles by the
synaptic connection between several neurons.}
\label{Synapsis}
\end{figure}

To test the influence of cycles on ES under degree-biased control
we simulate the effects of adding cycles to the Cayley irregular graph presented before. We have
studied the influence of two kind of cycles, small and large. For
small cycles we have considered triangles while for the large ones
we study the cycles of maximum length connecting a node of degree $1$ back with the node of maximum degree. The most remarkable
result obtained here is that when the connections between branched
trees create long cycles the synchronization process controlled by
the classical Laplacian continues to be explosive (see Fig.~\ref{BALong}, panel a). This means that in a process which is
not biased by degrees, the ES will remain after the single neuron
makes its synaptic connections. This will possibly always induce pathological
states, which is an antinatural result. Therefore, we may infer, or
suggest, that the process of synchronization is not controlled by
such Laplacian operator. On the contrary, in a process which is biased
by degrees we observe that ES drastically reduces as soon as long cycles
emerge. This is observed in Fig.~\ref{BALong} panels b and c, where in spite of the presence of certain
hysteresis, the first order transition has disappeared. We
can consider this process as a topologically-controlled transition
from explosive to classical synchronization, which may be used at
neuronal systems to avoid the appearance of pathological states after
the synaptic creation of neuronal networks. The situation is more
tricky to analyze in the case of small cycles. Although these cycles
are very hard to be formed during the synaptic connections of highly
branched neurons we explore their effects anyway. First, we observe
that in this case the ES disappears for both the degree-biased and the classical process, showing nevertheless some hysteresis as in the previous case (see Fig.~\ref{BAShortclas}, \ref{BAShortattract}).
It is important to remark that the hubs-repelling Laplacian behaves more erratic when this short cycle appears in the Cayley irregular graph (see Fig.~\ref{BAShortrepell}). 

\begin{figure*}[ht]
    \centering
    \subfigure[]{\includegraphics[width=0.32\linewidth]{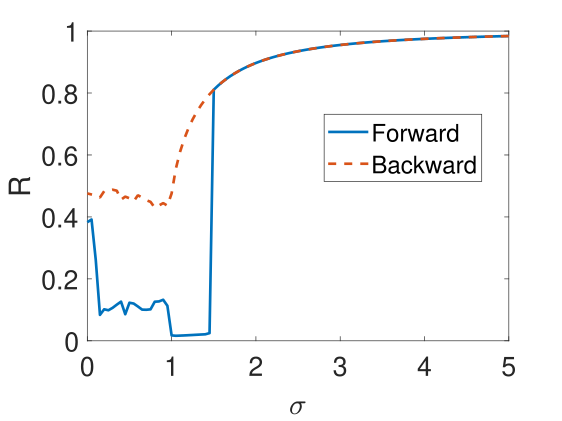}} 
    \subfigure[]{\includegraphics[width=0.32\linewidth]{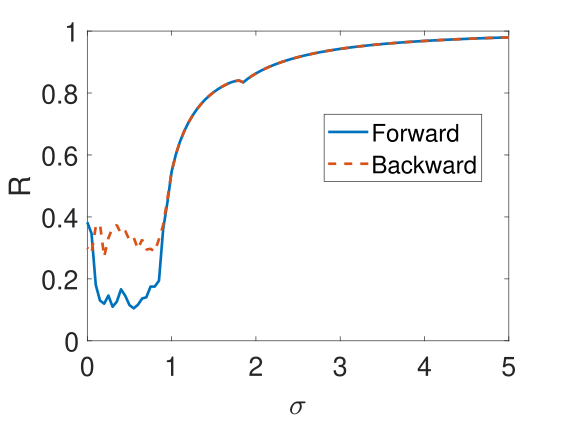}}
    \subfigure[]{\includegraphics[width=0.32\linewidth]{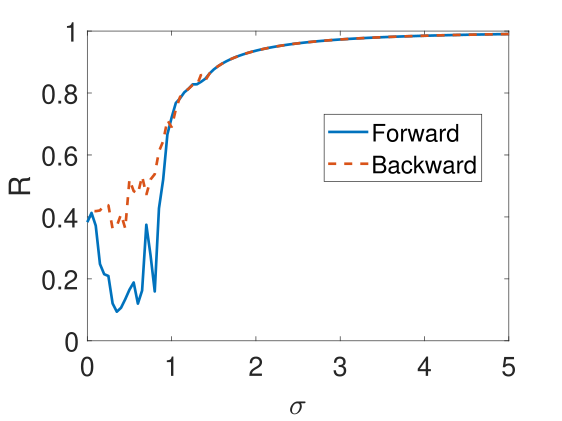}}
    \caption{ Kuramoto order parameter for the forward and backward continuations in a Cayley graph with a cycle of long length added for the classical (a), hubs-attracting (b) and hubs-repelling Laplacian (c).} \label{BALong}
\end{figure*}

\begin{figure*}[ht]
    \centering
    \subfigure[]{\includegraphics[width=0.32\linewidth]{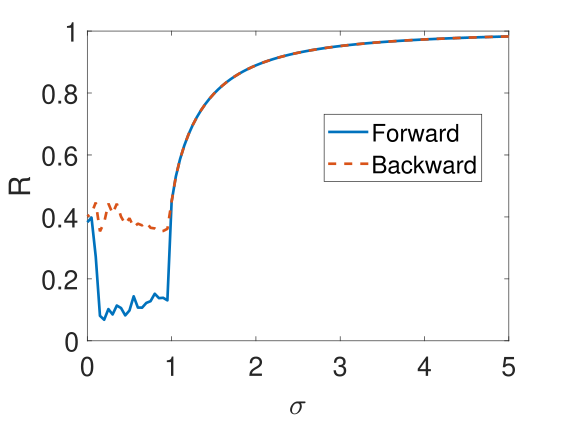} \label{BAShortclas} }
    \subfigure[]{\includegraphics[width=0.32\linewidth]{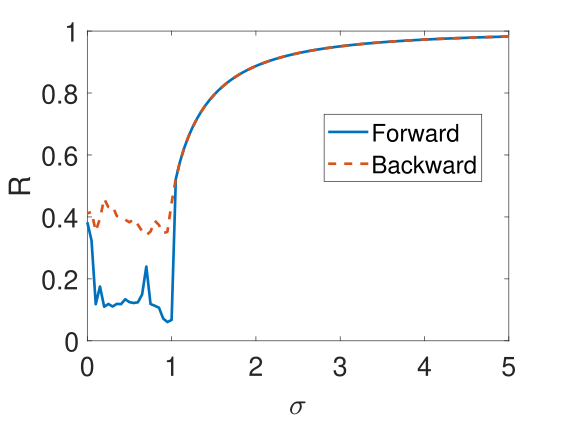} \label{BAShortattract} }
    \subfigure[]{\includegraphics[width=0.32\linewidth]{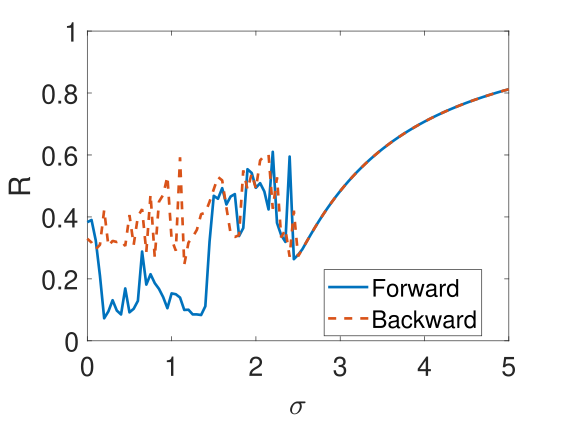} \label{BAShortrepell} }
    \caption{ Kuramoto order parameter for the forward and backward continuations in a Cayley graph with a cycle of small length added for the classical (a), hubs-attracting (b) and hubs-repelling Laplacian (c).} \label{BAShort}
\end{figure*}

Different types of neuron may have different shapes, closer to scale-free trees. These graphs, as shown in Fig.~\ref{BATree}, still show first order transitions when using any of the three Laplacians that we studied, with the only difference that they do not reach full synchronization due to the spatial asymmetry of the node degrees. Nevertheless, when including loops in this graphs, results are similar to those in Figs.~\ref{BALong} and \ref{BAShort}, so the ES disappears only for the hubs-biased Laplacians cases.

Based on the previous results we hypothesize here that degree-biased
control of ES could be a natural way of changing from the ES of a
single highly branched neuron to the standard synchronization of an
arrangement of neurons forming long cycles. The appearance of pathological
states emerges if the system lost its degree-biased control and pass
to a regime in which it is controlled by an operator like the classical
Laplacian. In favor of our hypotheses we can mention the result of
Münch and Werblin \cite{startburst_2}, who have found that the symmetric
interactions within a homogeneous starburst neuron network can lead
to robust asymmetries in dendrites of starburst amacrine cells. That
is, like in our case, the communication between neuronal cells could
be asymmetric in the sense that they could be biased by the degrees
of the nodes. This also coincides with strategies, like the one proposed
by Kaplan et al. \cite{Pain_2} consisting in targeting network hubs with
noninvasive brain stimulation to treat fibromyalgia in patients. In
any case our hypothesis is experimentally falsifiable to support or
reject it.

\section{Conclusions}
Explosive synchronization, in which networked Kuramoto oscillators
display an abrupt transition from a non-synchronized state to a synchronized
one, is certainly a very attractive phenomenon. Thinking to the phenomenon immediately brings to our mind the processing
of information by human brain. The human brain is able to take actions,
which require synchronization in the neuronal system, just in fractions
of a second. However, it seems that ES is not
beneficial for neuronal systems, but it could be causally-linked to
some brain pathologies. The most intriguing finding is that this first
order transition coexists with normal synchronization in a model like
the Kuramoto one. Then, how a neuronal system avoids entering into
pathological states by naturally suppressing ES?
We have shown here that a sort of topological control can emerge in
a neuronal-like system if the synchronization process is controlled
by a degree-biased diffusion. That is, by a diffusive process in which
a node passes its state to its nearest neighbors in an asymmetric
way, giving preference to those nodes with the highest (or lowest)
degrees. In this scenario, a branched tree displays ES,
which disappears as soon as it connect with other branched trees forming
cycles. This scenario is not allowed in the case in which the synchronization
is dominated by the classical process in which the state of a node
is passed to its nearest neighbors in a nonbiased way. This model
opens a new scenario for thinking about the potentially important
role that degree-biased ES and its topological
regulation by means of cycles create for neuronal systems.

\begin{acknowledgments}
M.M. thanks financial support PRE2020-092875 by MCIN/AEI /10.13039/501100011033 and by FSE invierte en tu futuro. E.E. thanks Grant PID2019-107603GB-I00
by MCIN/ AEI /10.13039/501100011033.
\end{acknowledgments}

\end{document}